\documentclass[12pt]{article}    

\usepackage[margin=0.75in]{geometry} 
\usepackage{amsmath,amssymb,amsthm}

\newtheorem{theorem}{Theorem}[section]
\newtheorem{proposition}[theorem]{Proposition}
\newtheorem{lemma}[theorem]{Lemma}
\newtheorem{corollary}[theorem]{Corollary}

\newtheorem{remark}{Remark}

\newcommand{\al}{\alpha}
\newcommand{\bt}{\beta}

\newcommand{\s}{\sigma}

\newcommand{\be}{\begin{equation}}
\newcommand{\ee}{\end{equation}}
\newcommand{\bea}{\begin{eqnarray}}
\newcommand{\eea}{\end{eqnarray}}

\numberwithin{equation}{section}
\begin{document}

\title{Painlev\'{e} V and the Hankel Determinant for a Singularly Perturbed Jacobi Weight}
\author{Chao Min and Yang Chen}


\date{\today}
\maketitle
\begin{abstract}
We study the Hankel determinant generated by a singularly perturbed Jacobi weight
$$
w(x,t):=(1-x^2)^\al\mathrm{e}^{-\frac{t}{x^{2}}},\;\;\;\;\;\;x\in[-1,1],\;\;\al>0,\;\;t\geq 0.
$$
If $t=0$, it is reduced to the classical symmetric Jacobi weight. For $t>0$, the factor $\mathrm{e}^{-\frac{t}{x^{2}}}$ induces an infinitely strong zero at the origin. This Hankel determinant is related to the Wigner time-delay distribution in chaotic cavities.

In the finite $n$ dimensional case, we obtain two auxiliary quantities $R_n(t)$ and $r_n(t)$ by using the ladder operator approach. We show that the Hankel determinant has an integral representation in terms of $R_n(t)$, where $R_n(t)$ is closely related to a particular Painlev\'{e} V transcendent. Furthermore, we derive a second-order nonlinear differential equation and also a second-order difference equation for the logarithmic derivative of the Hankel determinant. This quantity can be expressed in terms of the Jimbo-Miwa-Okamoto $\s$-function of a particular Painlev\'{e} V. Then we consider the asymptotics of the Hankel determinant under a suitable double scaling, i.e. $n\rightarrow\infty$ and $t\rightarrow 0$ such that $s=2n^2 t$ is fixed. Based on previous results by using the Coulomb fluid method, we obtain the large $s$ and small $s$ asymptotic behaviors of the scaled Hankel determinant, including the constant term in the asymptotic expansion.
\end{abstract}

$\mathbf{Keywords}$: Random matrix theory; Hankel determinant; Singularly perturbed Jacobi weight;

Ladder operators; Painlev\'{e} V; Double scaling.


\section{Introduction and Preliminaries}
Random matrices were introduced in nuclear physics by Wigner in the 1950s to describe the statistics of the energy levels of quantum systems. The theory of random matrices makes the hypothesis that the local statistical behavior of the energy levels is identical with that of the eigenvalues of a random matrix.

In random matrix theory (RMT), it is well known that the joint probability density of the eigenvalues $\{x_j\}_{j=1}^{n}$ of $n\times n$ Hermitian matrices in the unitary ensemble is \cite{Mehta}
$$
p(x_1, x_2,\ldots, x_n)dx_1 dx_2\cdots dx_n=\frac{1}{D_n[w]}\frac{1}{n!}\prod_{1\leq j<k\leq n}(x_j-x_k)^2\prod_{l=1}^{n}w(x_l)dx_l,
$$
where $w(x)$ is a weight or probability density supported on the interval $I\subseteq \mathbb{R}$, and we suppose the moments of all orders,
$$
\mu_{j}:=\int_{I}x^{j}w(x)dx,\qquad j=0,1,2,\ldots
$$
exist. In addition, $D_n[w]$ is the normalization constant or the partition function,
$$
D_n[w]=\frac{1}{n!}\int_{I^n}\prod_{1\leq j<k\leq n}(x_j-x_k)^2\prod_{l=1}^{n}w(x_l)dx_l.
$$
As a matter of fact, $D_n[w]$ can be expressed as the determinant of the Hankel or moment matrix \cite[(2.2.11)]{Szego}:
$$
D_{n}[w]=\det\left(\mu_{j+k}\right)_{j,k=0}^{n-1}.
$$

Hankel determinants have been studied extensively over the past few years in part due to connections with RMT \cite{BWW,Charlier,CG,CC,LCF,MLC,Wu,Xu2015,Xu2016}. This is because Hankel determinants compute the most fundamental objects in RMT, such as the partition function of a random matrix ensemble, the probability distribution function of the largest eigenvalue of Hermitian matrices and the moment generating function of linear statistics associated with the ensemble.

In this paper, we consider the Hankel determinant generated by a singularly perturbed Jacobi weight, namely,
$$
D_{n}(t):=\det\left(\int_{-1}^{1}x^{j+k}w(x,t)dx\right)_{j,k=0}^{n-1},
$$
where
\be\label{wei}
w(x,t):=(1-x^2)^\al\mathrm{e}^{-\frac{t}{x^{2}}},\;\;\;\;\;\;x\in[-1,1],\;\;\al>0,\;\;t\geq 0.
\ee
The study of this Hankel determinant is motivated in part by the Wigner time-delay distribution in chaotic cavities \cite{Texier}. In recent years, the asymptotic analysis of orthogonal polynomials and Hankel determinants for the singularly perturbed weights have attracted a lot of interests; see, e.g., \cite{BMM,CCF2019,MLC,Wang,Xu2015,Xu2016}.

Noting that the weight (\ref{wei}) is even, an evaluation of the moments shows that
\bea
\mu_{j}(t):&=&\int_{-1}^{1}x^{j}w(x,t)dx\nonumber\\[5pt]
&=&\left\{
\begin{aligned}
&0,&j=1,3,5,\ldots,\\
&(-1)^{\frac{j}{2}}\pi\left(\frac{\Gamma(1+\al)\Phi\left(-\frac{j+1}{2}-\al,\frac{1-j}{2};-t\right)}{\Gamma\left(\frac{1-j}{2}\right)\Gamma\left(\frac{j+3}{2}+\al\right)}
-\frac{t^{\frac{j+1}{2}}\Phi\left(-\al,\frac{j+3}{2};-t\right)}{\Gamma\left(\frac{j+3}{2}\right)}\right),&j=0,2,4,\ldots.
\end{aligned}
\right.\nonumber
\eea
where $\Phi(\cdot, \cdot; \cdot)$ is the confluent hypergeometric function \cite{Lebedev}.

In addition, for the unperturbed symmetric Jacobi weight, i.e. $t=0$, we have the explicit formula of $D_n(0)$ \cite[(17.6.2)]{Mehta}:
\bea\label{bg}
D_n(0)&=&\frac{2^{n(n+2\al)}}{n!}\prod_{j=1}^{n}\frac{\Gamma(j+1)\Gamma^2(j+\al)}{\Gamma(j+n+2\al)}\nonumber\\
&=&2^{n(n+2\al)}\frac{G(n+1)G(n+1+2\al)G^2(n+\al+1)}{G(2n+2\al+1)G^2(\al+1)},
\eea
where $G(\cdot)$ is the Barnes $G$-function which satisfies the relation \cite{Barnes}
$$
G(z+1)=\Gamma(z)G(z),\qquad\qquad G(1):=1.
$$
See also \cite{Voros} for more properties of this function.

It is well known that $D_n(t)$ can be expressed as the product of the square of the $L^2$ norms of the monic orthogonal polynomials \cite[(2.1.6)]{Ismail}, namely
\be\label{hankel}
D_{n}(t)=\prod_{j=0}^{n-1}h_{j}(t),
\ee
where
\be\label{or}
h_{j}(t)\delta_{jk}=\int_{-1}^{1}P_{j}(x,t)P_{k}(x,t)w(x,t)dx,\qquad j, k=0,1,2,\ldots.
\ee
and $P_{j}(x,t),\; j=0,1,2,\ldots,$ are the monic polynomials of degree $j$ orthogonal with respect to the weight function $w(x,t)$.

The main method of this paper is the ladder operator approach associated with the orthogonal polynomials. This approach has been widely applied to solve the problems on Hankel determinants generated by various weight functions; see \cite{ChenDai,ChenIts,Dai,Min2020,MLC} for reference. So we introduce some elementary facts about the orthogonal polynomials.

In the following discussions, $n$ denotes any nonnegative integers instead of the dimension of the Hermitian matrices. Since the weight $w(x,t)$ is even, $P_{n}(x,t)$ comprises only even or odd powers of $x$ when $n$ is even or odd respectively \cite{Chihara}. That is,
\be\label{expan}
P_{n}(x,t)=x^{n}+\mathrm{p}(n,t)x^{n-2}+\cdots,\qquad n=0,1,2,\ldots,
\ee
and we shall see that $\mathrm{p}(n,t)$, the coefficient of $x^{n-2}$, will play an important role. We set $\mathrm{p}(0,t)=0, \mathrm{p}(1,t)=0$.

The orthogonal polynomials $\{P_{n}(x,t)\}_{n=0}^{\infty}$ satisfy the three-term recurrence relation \cite{Szego}
\be\label{rr}
xP_{n}(x,t)=P_{n+1}(x,t)+\beta_{n}(t)P_{n-1}(x,t)
\ee
with the initial conditions
$$
P_{0}(x,t)=1,\;\;\;\;\beta_{0}(t)P_{-1}(x,t)=0.
$$
From (\ref{or}), (\ref{expan}) and (\ref{rr}), it is easy to see that $\bt_n(t)$ has two alternative representations:
\bea
\beta_{n}(t)&=&\mathrm{p}(n,t)-\mathrm{p}(n+1,t)\label{be1}\\
&=&\frac{h_{n}(t)}{h_{n-1}(t)},\label{be2}
\eea
and a telescopic sum gives
\be\label{sum}
\sum_{j=0}^{n-1}\beta_{j}(t)=-\mathrm{p}(n,t).
\ee
Note that from (\ref{be1}) we have $\bt_0(t)=0$.

The rest of this paper is arranged as follows. In Sec. 2, we apply the ladder operators and its supplementary conditions to the singularly perturbed Jacobi weight, and we obtain the second-order nonlinear difference equations satisfied by the auxiliary quantities $R_n(t)$ and $r_n(t)$, respectively. In Sec. 3, we show that $R_n(t)$ and $r_n(t)$ satisfy the coupled Riccati equations, from which we obtain the second-order differential equations satisfied by $R_n(t)$ and $r_n(t)$, respectively. We find that the equation for $R_n(t)$ can be transformed to a particular Painlev\'{e} V. In Sec. 4, we derive the second-order differential equation and also the difference equation satisfied by the logarithmic derivative of the Hankel determinant. We show that this quantity can be expressed as an integral in terms of $R_n(t)$. In Sec. 5, we study the asymptotic behavior of the Hankel determinant under a suitable double scaling, and the large $s$ and small $s$ asymptotic expansions of the scaled Hankel determinant are established.

\section{Ladder Operators and Second-Order Nonlinear Difference Equations}
In this section, we introduce a pair of ladder operators and the compatibility conditions at first. The information below can be found in Chen and Its \cite{ChenIts}.

Suppose that $w(x)$ is a continuous even weight function defined on $[-1,1]$, and $w(-1)=w(1)=0$. The monic orthogonal polynomials with respect to $w(x)$ satisfy the lowering operator equation
\be\label{lowering}
\left(\frac{d}{dz}+B_{n}(z)\right)P_{n}(z)=\beta_{n}A_{n}(z)P_{n-1}(z)
\ee
and the raising operator equation
\be\label{raising}
\left(\frac{d}{dz}-B_{n}(z)-\mathrm{v}'(z)\right)P_{n-1}(z)=-A_{n-1}(z)P_{n}(z),
\ee
where $\mathrm{v}(z):=-\ln w(z)$, and
\be\label{an}
A_{n}(z):=\frac{1}{h_{n}}\int_{-1}^{1}\frac{\mathrm{v}'(z)-\mathrm{v}'(y)}{z-y}P_{n}^{2}(y)w(y)dy,
\ee
\be\label{bn}
B_{n}(z):=\frac{1}{h_{n-1}}\int_{-1}^{1}\frac{\mathrm{v}'(z)-\mathrm{v}'(y)}{z-y}P_{n}(y)P_{n-1}(y)w(y)dy.
\ee

The functions $A_n(z)$ and $B_n(z)$ are not independent but must satisfy the following supplementary conditions valid for $z\in \mathbb{C}\cup\{\infty\}$:
\be
B_{n+1}(z)+B_{n}(z)=z A_{n}(z)-\mathrm{v}'(z), \tag{$S_{1}$}
\ee
\be
1+z(B_{n+1}(z)-B_{n}(z))=\beta_{n+1}A_{n+1}(z)-\beta_{n}A_{n-1}(z), \tag{$S_{2}$}
\ee
\be
B_{n}^{2}(z)+\mathrm{v}'(z)B_{n}(z)+\sum_{j=0}^{n-1}A_{j}(z)=\beta_{n}A_{n}(z)A_{n-1}(z). \tag{$S_{2}'$}
\ee

Solving for $P_{n-1}(z)$ from (\ref{lowering}) and substituting it into (\ref{raising}), we find that
$y(z):=P_{n}(z)$ satisfies the following second-order linear ordinary differential equation:
$$
y''(z)-\left(\mathrm{v}'(z)+\frac{A_{n}'(z)}{A_{n}(z)}\right)y'(z)+\left(B_{n}'(z)-B_{n}(z)\frac{A_{n}'(z)}{A_{n}(z)}
+\sum_{j=0}^{n-1}A_{j}(z)\right)y(z)=0,
$$
where use has been made of ($S_{2}'$).

In our problem, the weight and associated quantities are
$$
w(x,t)=(1-x^2)^\al\mathrm{e}^{-\frac{t}{x^{2}}},
$$
$$
\mathrm{v}(z)=\frac{t}{z^{2}}-\al\ln(1-z^2),\qquad\mathrm{v}'(z)=-\frac{2t}{z^3}+\frac{2\al z}{1-z^2},
$$
$$
\frac{\mathrm{v}'(z)-\mathrm{v}'(y)}{z-y}=\frac{2t(z^2+zy+y^2)}{z^3y^3}+\frac{2\al(1+zy)}{(1-z^2)(1-y^2)}.
$$
In the following discussions, we suppress the $t$ dependence in $w(x), h_n$ and $\bt_n$ for brevity.
\begin{proposition}
For our problem, we have
\be\label{anz}
A_{n}(z)=\frac{R_{n}(t)}{z^2}+\frac{2n+1+2\al+R_{n}(t)}{1-z^2},
\ee
\be\label{bnz}
B_{n}(z)=\frac{\left[1-(-1)^n\right]t}{z^3}+\frac{r_{n}(t)}{z}+\frac{z(n+r_n(t))}{1-z^2},
\ee
where
\be\label{Rnt}
R_{n}(t):=\frac{2t}{h_{n}}\int_{-1}^{1}\frac{P_{n}^{2}(y)w(y)}{y^2} dy,
\ee
\be\label{rnt}
r_{n}(t):=\frac{2t}{h_{n-1}}\int_{-1}^{1}\frac{P_{n}(y)P_{n-1}(y)w(y)}{y^3} dy.
\ee
\end{proposition}

\begin{proof}
According to the definition of $A_{n}(z)$ in (\ref{an}), we get
\bea
A_{n}(z)
&=&\frac{1}{h_{n}}\int_{-1}^{1}\left[\frac{2t(z^2+zy+y^2)}{z^3y^3}+\frac{2\al(1+zy)}{(1-z^2)(1-y^2)}\right]P_{n}^{2}(y)w(y)dy\nonumber\\[5pt]
&=&\frac{2t}{z^2h_{n}}\int_{-1}^{1}\frac{P_{n}^{2}(y)w(y)}{y^2} dy+\frac{2\al}{(1-z^2)h_{n}}\int_{-1}^{1}\frac{P_{n}^{2}(y)w(y)}{1-y^2} dy,\nonumber
\eea
where we have used the parity of the integrand to simplify the result. The above two integrals have a simple relation through integration by parts,
\bea
\frac{2t}{h_{n}}\int_{-1}^{1}\frac{P_{n}^{2}(y)w(y)}{y^2} dy&=&\frac{1}{h_{n}}\int_{-1}^{1}y P_{n}^{2}(y)(1-y^2)^\al d\mathrm{e}^{-\frac{t}{y^{2}}}\nonumber\\
&=&-2n-1+\frac{2\al}{h_{n}}\int_{-1}^{1}y^2P_{n}^{2}(y)(1-y^2)^{\al-1}\mathrm{e}^{-\frac{t}{y^{2}}}dy\nonumber\\
&=&-2n-1-2\al+\frac{2\al}{h_{n}}\int_{-1}^{1}\frac{P_{n}^{2}(y)w(y)}{1-y^2} dy.\nonumber
\eea
That is,
$$
\frac{2\al}{h_{n}}\int_{-1}^{1}\frac{P_{n}^{2}(y)w(y)}{1-y^2} dy=2n+1+2\al+R_n(t).
$$
Then formula (\ref{anz}) follows.

Similarly, from the definition of $B_n(z)$ in (\ref{bn}), we have
\bea
B_{n}(z)&=&\frac{1}{h_{n-1}}\int_{-1}^{1}\left[\frac{2t(z^2+zy+y^2)}{z^3y^3}+\frac{2\al(1+zy)}{(1-z^2)(1-y^2)}\right]P_{n}(y)P_{n-1}(y)w(y)dy\nonumber\\[5pt]
&=&\frac{2t}{z^3h_{n-1}}\int_{-1}^{1}\frac{P_{n}(y)P_{n-1}(y)w(y)}{y} dy+\frac{2t}{zh_{n-1}}\int_{-1}^{1}\frac{P_{n}(y)P_{n-1}(y)w(y)}{y^3} dy\nonumber\\[5pt]
&+&\frac{2\al z}{(1-z^2)h_{n-1}}\int_{-1}^{1}\frac{yP_{n}(y)P_{n-1}(y)w(y)}{1-y^2} dy.\nonumber
\eea
In view of (\ref{or}) and (\ref{expan}), it follows that
\bea\label{gf}
\frac{1}{h_{n-1}}\int_{-1}^{1}\frac{P_{n}(y)P_{n-1}(y)w(y)}{y} dy&=&\left\{
\begin{aligned}
&0,&n=0,2,4,\ldots,\\
&1,&n=1,3,5,\ldots
\end{aligned}
\right.\nonumber\\[5pt]
&=&\frac{1-(-1)^n}{2}.
\eea
By integration by parts, we have
\bea
\frac{2\al}{h_{n-1}}\int_{-1}^{1}\frac{yP_{n}(y)P_{n-1}(y)w(y)}{1-y^2} dy&=&-\frac{1}{h_{n-1}}\int_{-1}^{1}P_{n}(y)P_{n-1}(y)\mathrm{e}^{-\frac{t}{y^2}} d(1-y^2)^\al\nonumber\\[5pt]
&=&\frac{1}{h_{n-1}}\int_{-1}^{1}(1-y^2)^\al \left(P_{n}(y)P_{n-1}(y)\mathrm{e}^{-\frac{t}{y^2}}\right)' dy\nonumber\\[5pt]
&=&n+\frac{2t}{h_{n-1}}\int_{-1}^{1}\frac{P_{n}(y)P_{n-1}(y)w(y)}{y^3} dy\nonumber\\[5pt]
&=&n+r_n(t).\nonumber
\eea
Then we arrive at formula (\ref{bnz}). This completes the proof.
\end{proof}

Substituting (\ref{anz}) and (\ref{bnz}) into ($S_{1}$), we find
\be\label{s1}
R_{n}(t)=r_{n+1}(t)+r_{n}(t).
\ee
Similarly, substituting (\ref{anz}) and (\ref{bnz}) into ($S_{2}$) gives the following two equations:
\be\label{s21}
\bt_{n+1}R_{n+1}(t)-\bt_{n}R_{n-1}(t)=2(-1)^nt,
\ee
\be\label{s22}
\bt_{n+1}(2n+3+2\al+R_{n+1}(t))-\bt_{n}(2n-1+2\al+R_{n-1}(t))=1+r_{n+1}(t)-r_n(t).
\ee
The combination of (\ref{s21}) and (\ref{s22}) shows
$$
2(-1)^nt+(2n+3+2\al)\bt_{n+1}-(2n-1+2\al)\bt_{n}=1+r_{n+1}(t)-r_n(t),
$$
i.e.,
$$
2(-1)^nt+\left[(2n+1+2\al)\bt_{n+1}-(2n-1+2\al)\bt_{n}\right]+2\bt_{n+1}=1+r_{n+1}(t)-r_n(t).
$$
A telescopic sum gives
\be\label{s2}
\left[1-(-1)^n\right]t+(2n+1+2\al)\bt_n-2\mathrm{p}(n,t)=n+r_n(t),
\ee
where we have used (\ref{sum}) and the initial conditions $\bt_0=0, r_0(t)=0$.

Finally, from ($S_{2}'$), we obtain the following three equations:
\be\label{s2p1}
-2(-1)^n t\: r_{n}(t)=\beta_{n}R_{n}(t)R_{n-1}(t),
\ee
\be\label{s2p2}
(n+r_{n}(t))^2+2\al(n+r_{n}(t))=\beta_{n}\left(2n+1+2\al+R_{n}(t)\right)\left(2n-1+2\al+R_{n-1}(t)\right),
\ee
\bea\label{s2p3}
&&r_{n}^2(t)-2(-1)^n t(n+r_{n}(t))+2\al\left[1-(-1)^n\right]t+\sum_{j=0}^{n-1}R_{j}(t)\nonumber\\
&=&\beta_{n}\left(2n+1+2\al+R_{n}(t)\right)R_{n-1}(t)+\beta_{n}R_{n}(t)\left(2n-1+2\al+R_{n-1}(t)\right).
\eea
The combination of (\ref{s2p1}) and (\ref{s2p2}) gives an important formula:
\bea\label{imp}
&&(2n+1+2\al)\bt_nR_{n-1}(t)+(2n-1+2\al)\bt_nR_{n}(t)\nonumber\\
&=&(n+r_{n}(t))^2+2\al(n+r_{n}(t))+2(-1)^n t\: r_n(t)-(2n+1+2\al)(2n-1+2\al)\bt_n.
\eea
Furthermore, substituting $\bt_nR_{n-1}(t)=-2(-1)^n t\: r_{n}(t)/R_n(t)$ into (\ref{imp}), we obtain the expression of $\bt_n$ in terms of $R_n(t)$ and $r_n(t)$:
\be\label{be3}
\bt_n=\frac{(n+r_{n}(t))^2+2\al(n+r_{n}(t))}{(2n-1+2\al)(2n+1+2\al+R_{n}(t))}+\frac{2(-1)^n t\: r_{n}(t)}{(2n-1+2\al)R_{n}(t)}.
\ee
\begin{theorem}
The auxiliary quantities $r_n(t)$ and $R_n(t)$ satisfy the following second-order difference equations:
\bea\label{rn}
&&(n+r_{n})(n+2\al+r_{n})(r_{n+1}+r_{n})(r_n+r_{n-1})\nonumber\\
&+&2(-1)^n t\: r_{n}(2n-1+2\al+r_n+r_{n-1})(2n+1+2\al+r_{n+1}+r_{n})=0,
\eea\\[-40pt]
\begin{small}
\bea\label{Rn}
&&\Big\{\big[R_{n+1}R_n^4+ 2\big((2 n+2 \alpha+1) R_{n+1}- (-1)^nt(2 n+2 \alpha+3)\big)R_n^3+\big( (4 n^2+4n(2 \alpha +1)\nonumber\\
&-&2 t^2+2 (-1)^n t+2 \alpha ^2+4 \alpha+1)R_{n+1}-2 t (2 n+2 \alpha+3) (t+(-1)^n(3n+3 \alpha +1))\big)R_n^2\nonumber\\
&+&2 t (2 n+2 \alpha+1)\big( ((-1)^n-2 t)R_{n+1}-(2 n+2 \alpha+3) (2 t+(-1)^n (n+\al))\big) R_n\nonumber\\
&-&2 t^2 (2 n+2 \alpha+1)^2 (2n+2\al+3+R_{n+1})\big]R_{n-1}+2 t (2 n+2 \alpha-1) (2n+2 \alpha+1 +R_n)\nonumber\\
&\times&\big[(-1)^n R_{n+1}R_n^2+ \big( ((-1)^n (n+\al+1)-t)R_{n+1}-t (2 n+2 \alpha+3)\big)R_n-t (2 n+2 \alpha+1)\nonumber\\
&\times&(2n+2 \alpha+3 +R_{n+1})\big]\Big\}^2=4\Big\{\big[(-1)^n t (2 n+2 \alpha-1) (2n+2 \alpha+1 +R_n)\nonumber\\
&+& \big( (n+\alpha +(-1)^n t)R_n+(-1)^n t (2 n+2 \alpha+1)\big)R_{n-1}\big]^2- n (n+2 \alpha) R_{n-1}^2 R_n^2\Big\}\nonumber\\
&\times&\Big\{\big[(-1)^n t (2 n+2 \alpha+1) (2n+2 \alpha+3 +R_{n+1})- \big( (n+\alpha+1 -(-1)^n t)R_{n+1}\nonumber\\
&-&(-1)^n t (2 n+2 \alpha+3)\big)R_{n}\big]^2-(n+1) (n+2 \alpha+1) R_n^2 R_{n+1}^2\Big\}.
\eea
\end{small}
The initial conditions are given by
$$
r_0(t)=0,\\[10pt]
$$
\be\label{ic1}
r_1(t)=R_0(t)=\frac{2\sqrt{t}\left[\Gamma\left(\frac{3}{2}+\al\right)\Phi\left(-\al,\frac{1}{2};-t\right)-(1+2\al)\sqrt{t}\:\Gamma(1+\al)\Phi\left(\frac{1}{2}-\al,\frac{3}{2};-t\right)
\right]}{\Gamma(1+\al)\Phi\left(-\frac{1}{2}-\al,\frac{1}{2};-t\right)-2\sqrt{t}\:\Gamma\left(\frac{3}{2}+\al\right)\Phi\left(-\al,\frac{3}{2};-t\right)},\\[10pt]
\ee
\be\label{ic2}
R_1(t)=\frac{6t\left[(3+2\al)\Gamma(1+\al)\Phi\left(-\frac{1}{2}-\al,\frac{1}{2};-t\right)-4\sqrt{t}\:\Gamma\left(\frac{5}{2}+\al\right)\Phi\left(-\al,\frac{3}{2};-t\right)
\right]}{3\Gamma(1+\al)\Phi\left(-\frac{3}{2}-\al,-\frac{1}{2};-t\right)+8t\sqrt{t}\:\Gamma\left(\frac{5}{2}+\al\right)\Phi\left(-\al,\frac{5}{2};-t\right)},
\ee
where $\Phi(\cdot, \cdot; \cdot)$ is the confluent hypergeometric function.
\end{theorem}
\begin{proof}
Substituting (\ref{s1}) into (\ref{s2p1}), we have
\be\label{d1}
\bt_n=-\frac{2(-1)^n t\: r_{n}(t)}{(r_{n+1}(t)+r_{n}(t))(r_{n}(t)+r_{n-1}(t))}.
\ee
The combination of (\ref{d1}) and (\ref{be3}) gives (\ref{rn}) with the aid of (\ref{s1}). On the other hand, eliminating $\bt_n$ from (\ref{s2p1}) and (\ref{s2p2}), we find
\bea
&&\left[(n+r_{n}(t))^2+2\al(n+r_{n}(t))\right]R_{n}(t)R_{n-1}(t)\nonumber\\
&+&2(-1)^n t\: r_n(t)\left(2n+1+2\al+R_{n}(t)\right)\left(2n-1+2\al+R_{n-1}(t)\right)=0.\nonumber
\eea
This equation can be viewed as the qradratic equation for $r_n(t)$. Substituting either solution into (\ref{s1}) and after clearing the square root, we obtain (\ref{Rn}). The initial conditions follow from straightforward computations of (\ref{Rnt}) and (\ref{rnt}) when $n=0$ and 1.
\end{proof}

\section{Coupled Riccati Equations, Painlev\'{e} V}
In this section, we derive the coupled Riccati equations satisfied by the auxiliary quantities $R_n(t)$ and $r_n(t)$ from taking derivatives in the orthogonality conditions. Then we obtain the second-order differential equations for $R_n(t)$ and $r_n(t)$, respectively. We find that $R_n(t)$ is closely related to a particular Painlev\'{e} V transcendent.

We start from taking a derivative with respect to $t$ in the equality
$$
\int_{-1}^{1}P_{n}^2(x,t)(1-x^2)^\al\mathrm{e}^{-\frac{t}{x^{2}}}dx=h_{n}(t),\;\;n=0,1,2,\ldots,
$$
and have
$$
h_n'(t)=-\int_{-1}^{1}\frac{P_{n}^2(x,t)w(x)}{x^2}dx.
$$
It follows that
\be\label{eq1}
2t \frac{d}{dt}\ln h_{n}(t)=-R_{n}(t).
\ee
In view of (\ref{be2}), we find
$$
2t\frac{d}{dt}\ln\beta_{n}(t)=R_{n-1}(t)-R_{n}(t).
$$
Then
\be\label{eq2}
2t\beta_{n}'(t)=\beta_{n}R_{n-1}(t)-\beta_{n}R_{n}(t).
\ee

On the other hand, taking a derivative in the equation
$$
\int_{-1}^{1}P_{n}(x,t)P_{n-2}(x,t)(1-x^2)^\al\mathrm{e}^{-\frac{t}{x^{2}}}dx=0,\;\;n=1,2,\ldots,
$$
we have
\be\label{pn}
\frac{d}{dt}\mathrm{p}(n,t)=\frac{1}{h_{n-2}}\int_{-1}^{1}\frac{P_{n}(x,t)P_{n-2}(x,t)w(x)}{x^2}dx.
\ee
From the three-term recurrence relation (\ref{rr}) and exploiting (\ref{be2}), we get
\be\label{rr1}
\frac{P_{n-2}(x,t)}{h_{n-2}}=\frac{xP_{n-1}(x,t)}{h_{n-1}}-\frac{P_{n}(x,t)}{h_{n-1}}.
\ee
Substituting (\ref{rr1}) into (\ref{pn}), we obtain
\be\label{pnt}
2t\frac{d}{dt}\mathrm{p}(n,t)=\left[1-(-1)^n\right]t-\beta_{n}R_{n}(t),
\ee
where we have used formula (\ref{gf}).

\begin{proposition} The auxiliary quantities $r_n(t)$ and $R_n(t)$ satisfy the coupled Riccati equations:
\be\label{ricca1}
2t\:r_{n}'(t)=-\frac{2(-1)^nt\:r_n(t) (2n+1+2\al+R_n(t))}{R_{n}(t)}-\frac{(n+r_{n}(t))(n+2\al+r_{n}(t))R_{n}(t)}{2n+1+2\al+R_{n}(t)},
\ee
\be\label{ricca2}
2tR_{n}'(t)=R_{n}^2(t)+\left[1-2(-1)^n t-2r_{n}(t)\right]R_{n}(t)-2(-1)^n(2n+2\al+1) t.
\ee
\end{proposition}

\begin{proof}
By taking a derivative in (\ref{s2}) and with the aid of (\ref{pnt}), we have
\be\label{for}
2t\: r_n'(t)=2(2n+1+2\al)t\bt_n'(t)+2\bt_nR_n(t).
\ee
Replacing $2t\bt_n'(t)$ by (\ref{eq2}), equation (\ref{for}) becomes
\bea
2t\:r_n'(t)&=&(2n+1+2\al)\bt_nR_{n-1}(t)-(2n-1+2\al)\bt_nR_{n}(t)\nonumber\\
&=&-\frac{2(-1)^n (2n+1+2\al)t\:r_n(t)}{R_{n}(t)}-(2n-1+2\al)\bt_nR_{n}(t).\nonumber
\eea
where we have used (\ref{s2p1}) in the second step. Substituting the expression of $\bt_n$ in (\ref{be3}) into the above equality, we obtain (\ref{ricca1}).

On the other hand, substituting (\ref{be3}) into (\ref{for}) to eliminate $\bt_n$ and $\bt_n'(t)$, and then replacing $2t\: r_n'(t)$ by (\ref{ricca1}), we find
\bea
&&\Big\{2tR_{n}'(t)-R_{n}^2(t)-\left[1-2(-1)^n t-2r_{n}(t)\right]R_{n}(t)+2(-1)^n(2n+2\al+1) t\Big\}\nonumber\\
&\times&\Big\{ \left[r_n^2(t)+2  \left(n+\alpha +(-1)^n t\right)r_n(t)+n (n+2 \alpha)\right]R_n^2(t)+4 (-1)^n t (2 n+2 \alpha +1) r_n(t) R_n(t)\nonumber\\
&&+2 (-1)^n t (2 n+2 \alpha +1)^2 r_n(t)\Big\}=0.\nonumber
\eea
Obviously, the above formula yields two equations. The latter algebraic equation does not hold and should be discarded, which can be verified by taking special values such as $n=1$. Hence we obtain (\ref{ricca2}). The proof is complete.
\end{proof}

\begin{theorem}
The auxiliary quantity $R_n(t)$ satisfies the following second-order nonlinear ordinary differential equation,
\bea\label{ode}
&&8 t^2 R_n(t)  (2 n+2 \alpha+1+R_n(t))R_n''(t)-4 t^2  (4 n+4 \alpha +2+3 R_n(t))(R_n'(t))^2\nonumber\\
&+&8 t  (2 n+2 \alpha +1+R_n(t))R_n(t) R_n'(t)-R_n^5(t)-2 (2 n+2 \alpha+1) R_n^4(t)\nonumber\\
&-& \left[4 n^2+4(2 \alpha +1) n+4 \alpha +1-4 t^2-4 (-1)^n t\right]R_n^3(t)+8 t (2 n+2 \alpha +1) \left[(-1)^n+2 t\right] R_n^2(t)\nonumber\\
&+&4 t (2 n+2 \alpha +1)^2 \left[(-1)^n+5 t\right] R_n(t)+8 t^2 (2 n+2 \alpha +1)^3=0.
\eea
Let $S_n(t):=1+\frac{R_{n}(t)}{2n+2\al+1}$. Then $S_n(t)$ satisfies a particular Painlev\'{e} V equation,
\bea\label{pv}
S_{n}''(t)&=&\frac{(3 S_n(t)-1) (S_n'(t))^2}{2S_n(t) (S_n(t)-1) }-\frac{S_n'(t)}{t}+\frac{(S_n(t)-1)^2 }{t^2}\left(\frac{(2n+2 \alpha+1)^2 S_n(t)}{8} -\frac{\al^2}{2S_n(t)}\right)\nonumber\\[10pt]
&-&\frac{(-1)^n S_n(t)}{2t}-\frac{S_n(t)(S_n(t)+1) }{2(S_n(t)-1)},
\eea
i.e. $P_{V}\left(\frac{(2n+2 \alpha+1)^2}{8}, -\frac{\al^2}{2}, -\frac{(-1)^n}{2}, -\frac{1}{2}\right)$, following the convention of \cite{Gromak}.
\end{theorem}
\begin{proof}
Note that equation (\ref{ricca2}) can be viewed as a linear equation for $r_n(t)$. Solving $r_{n}(t)$ from (\ref{ricca2}) we have
\be\label{rexp}
r_n(t)=\frac{1-2 (-1)^n t}{2}+\frac{R_n(t)}{2}-\frac{t \left[(-1)^n (2 n+1+2 \alpha)+R_n'(t)\right]}{R_n(t)}.
\ee
Substituting it into (\ref{ricca1}), we obtain (\ref{ode}). With the linear transformation $S_n(t)=1+\frac{R_{n}(t)}{2n+2\al+1}$, we obtain the Painlev\'{e} V equation (\ref{pv}).
\end{proof}

\begin{theorem}
The quantity $r_n(t)$ satisfies the second-order nonlinear differential equation:
\begin{small}
\bea\label{ode2}
&&4 t^3 (r_n''(t))^2+4t^2\big[ r_n'(t)-2(-1)^n \left(3r_n^2(t)+4(n+\alpha) r_n(t)+n  (n+2 \alpha)\right)\big]r_n''(t)\nonumber\\
&-&t\big[4r_n^2(t)-8 (-1)^n t r_n(t)+4 t^2-1\big] (r_n'(t))^2-4(-1)^nt\big[3r_n^2(t)+4   (n+\alpha)r_n(t)+ n  (n+2 \alpha)\big]r_n'(t)\nonumber\\
&+&8 (-1)^n r_n^5(t)+4  \big[4 (-1)^n (n+\al)+5 t\big]r_n^4(t)+8\big[(-1)^n (n^2+2n\alpha+t^2)+8 t (n+\alpha)\big]r_n^3(t)\nonumber\\
&+&8t\big[2(-1)^nt(n+\al)+(3n+2\al)(3n+4\al)\big]r_n^2(t)+8 n (n+2 \alpha)t \big[4 n+4 \alpha +(-1)^n t\big]r_n(t)\nonumber\\
&+&4 n^2 (n+2 \alpha)^2t=0.
\eea
\end{small}
\end{theorem}
\begin{proof}
Solving $R_n(t)$ from (\ref{ricca1}), which is actually a quadratic equation for $R_n(t)$, we get two solutions. Substituting either solution into (\ref{ricca2}), and clearing the square root, we obtain (\ref{ode2}) after making some simplifications.
\end{proof}

\begin{remark}
The second-order differential equation (\ref{ode2}) for $r_n(t)$ may be transformed to a Chazy type equation \cite{Chazy1909,Chazy1911,Cosgrove}. See \cite{Min2019,Lyu2017} on Chazy equations satisfied by $r_n(t)$ in other problems.
\end{remark}

\section{Logarithmic Derivative of the Hankel Determinant}
In this section, we study the logarithmic derivative of the Hankel determinant. Usually this quantity satisfies a second-order differential equation, which is related to the Jimbo-Miwa-Okamoto $\s$-form of a particular Painlev\'{e} equation.

We introduce a quantity related to the logarithmic derivative of the Hankel determinant,
\be\label{def}
\s_n(t):=2t\frac{d}{dt}\ln D_n(t).
\ee
In view of (\ref{hankel}) and (\ref{eq1}), we find
\be\label{sn}
\s_n(t)=-\sum_{j=0}^{n-1}R_j(t).
\ee
Then equation (\ref{s2p3}) becomes
\be\label{re}
(2n+1+2\al)\bt_nR_{n-1}(t)+(2n-1+2\al)\bt_nR_{n}(t)=2\al t-2(-1)^n t(n+\al-r_n(t))+r_n^2(t)-\s_n(t),
\ee
where we have used (\ref{s2p1}).

The combination of (\ref{imp}) and (\ref{re}) shows that $\bt_n$ can be expressed in terms of $r_n(t)$ and $\s_n(t)$:
\be\label{be4}
\bt_n=\frac{n(n+2\al)+2(n+\al)r_n(t)+\s_n(t)+2t\left[(-1)^n(n+\al)-\al\right]}{(2n+1+2\al)(2n-1+2\al)}.
\ee
By taking a derivative in (\ref{be4}) and substituting it into (\ref{eq2}), we find
\be\label{re2}
\beta_{n}R_{n-1}(t)-\beta_{n}R_{n}(t)=\frac{2 t \left[2(n+\al)((-1)^n+r_n'(t))+\sigma_n'(t)-2 \alpha\right]}{(2n+1+2\al)(2n-1+2\al)}.
\ee

\begin{theorem}
The quantity $\s_n(t)$ satisfies the following second-order nonlinear ordinary differential equation:
\begin{small}
\bea\label{sode}
&&\Big\{4 t^3 (\sigma _n'')^2- 4 t^2\left[ 2t+2\alpha -2 (-1)^n(n+\al)-\sigma _n'\right]\sigma _n''+8 t^2 (\sigma _n')^3-t\big[4 t^2+40 \alpha\:  t-1+4 \sigma _n\nonumber\\
&+&24 (-1)^n t (n+\alpha)\big](\sigma _n')^2+4t\big[12 \alpha \: t^2 -(20 n^2+40 n\alpha+3) t-\alpha+(-1)^n(n+\al) (12 t^2+1)\nonumber\\
&+&4  \left(\alpha-t +3 (-1)^n(n+\al)\right)\sigma _n\big]\sigma _n'+8 [t-2 (-1)^n(n+\al)]\sigma _n^2+4t\big[2 t^2+1+14 n^2+28   n\alpha\nonumber\\
&+&8 \alpha ^2-4(-1)^n (n+\alpha) (3 t+2 \alpha)\big]\sigma _n-4t\big[4 \alpha\:  t^3- 2t^2(7 n^2+14   n\al+1)-4 \alpha t (3 n^2+6   n\alpha+1)\nonumber\\
&-&n^2-2   n\alpha-2 \alpha ^2+2(-1)^n\left(2t^3 ( n+\al)+2\left(3 n^3+9   n^2\alpha+ n(6 \alpha ^2+1)+\alpha\right)t +n\alpha  +\alpha ^2\right)\big]\Big\}^2\nonumber\\
&=&16 \Big\{t[t+2 \alpha-2 \sigma _n' -2 (-1)^n(n+\al)]+\sigma _n\Big\}\Big\{2 (-1)^n t^2 \sigma _n''+t\big[(-1)^n(4 n^2+8n\al-8\alpha t+1+4 \sigma _n)\nonumber\\
&-&8t(n+\al)\big]\sigma _n'-2 (-1)^n \sigma _n^2+2\left[4t(n+\al)-(-1)^n(n^2+2 n\alpha  +t^2)\right]\sigma _n+2t\big[(n+\alpha )(2 n^2+4  n\alpha\nonumber\\
&+&2 t^2+1)+(-1)^n\left(2 \alpha\: t^2-(5 n^2+10   n\alpha+1)t-\alpha  (2 n^2+4   n\alpha+1)\right)\big]\Big\}^2.
\eea
\end{small}
\end{theorem}
\begin{proof}
We define
$$
X:=\bt_nR_{n-1}(t),
$$
$$
Y:=\bt_nR_{n}(t).
$$
Then $X$ and $Y$ satisfy a linear system from (\ref{re}) and (\ref{re2}):
$$
(2n+1+2\al)X+(2n-1+2\al)Y=2\al t-2(-1)^n t(n+\al-r_n(t))+r_n^2(t)-\s_n(t),
$$
$$
X-Y=\frac{2 t \left[2(n+\al)((-1)^n+r_n'(t))+\sigma_n'(t)-2 \alpha\right]}{(2n+1+2\al)(2n-1+2\al)}.
$$
Solving this linear system, we get
\bea\label{x}
X&=&\frac{2 \alpha  t-2 (-1)^n t (n+\alpha-r_n(t))+r_n^2(t)-\sigma_n (t)}{4(n+\al)}\nonumber\\
&+&\frac{t\left[2(n+\al)\left((-1)^n+r_n'(t)\right)+\s_n'(t)-2\al\right]}{2(n+\al)(2n+1+2\al)},
\eea
\bea\label{y}
Y&=&\frac{2 (-1)^n t\: r_n(t)+r_n^2(t)-\sigma_n (t)}{4 (n+\alpha)}\nonumber\\
&-&\frac{ t \left[(2 n+1+2 \alpha) \left((-1)^n (n+\al)-\al\right)+2 (n+\alpha) r_n'(t)+\sigma_n'(t)\right]}{2 (n+\alpha) (2 n-1+2 \alpha)}.
\eea
Noting that from the definition of $X$ and $Y$, we have
\bea\label{xy}
X\cdot Y&=&\bt_n(\bt_nR_{n}(t)R_{n-1}(t))\nonumber\\
&=&-2(-1)^n t\bt_n r_{n}(t).
\eea
where we have used (\ref{s2p1}).

Substituting (\ref{x}), (\ref{y}) and (\ref{be4}) into (\ref{xy}), we obtain an equation satisfied by $r_n(t), r_n'(t), \s_n(t)$ and $\s_n'(t)$:
\begin{small}
\bea\label{com}
&&\big\{(2n+1+2\al)\left[2 \alpha  t-2 (-1)^n t (n+\alpha-r_n(t))+r_n^2(t)-\sigma_n (t)\right]\nonumber\\
&+&2t\left[2(n+\al)\left((-1)^n+r_n'(t)\right)+\s_n'(t)-2\al\right]\big\}\big\{(2n-1+2\al)\left[2 (-1)^n t\: r_n(t)+r_n^2(t)-\sigma_n (t)\right]\nonumber\\
&-&2t \left[(2 n+1+2 \alpha) \left((-1)^n (n+\al)-\al\right)+2 (n+\alpha) r_n'(t)+\sigma_n'(t)\right]\big\}\nonumber\\
&=&-32(-1)^n(n+\al)^2 t\:r_n(t)\big\{n(n+2\al)+2(n+\al)r_n(t)+\s_n(t)+2t\left[(-1)^n(n+\al)-\al\right]\big\}.
\eea
\end{small}\\[-40pt]

On the other hand, the combination of (\ref{s2}) and (\ref{be4}) shows that $\mathrm{p}(n,t)$ can be expressed in terms of $r_n(t)$ and $\s_n(t)$:
$$
\mathrm{p}(n,t)=\frac{n-n^2+\left[2 n-1+(-1)^n\right] t+r_n(t)+\sigma_n (t)}{2 (2 n-1+2 \alpha)}.
$$
Taking a derivative and substituting it into (\ref{pnt}), and replacing $\bt_n R_n(t)$ by (\ref{y}), we find a quadratic equation for $r_n(t)$ (the terms involving $r_n'(t)$ disappear!):
$$
r_n^2(t)+2 (-1)^n t\: r_n(t)-\sigma_n (t)+2 t \left[(-1)^n(n+\al)-\alpha +\sigma_n'(t)\right]=0.
$$
Then we have two solutions,
$$
r_n(t)=-(-1)^n t\pm\sqrt{\sigma_n(t)+t \left[t+2 \alpha -2 (-1)^n(n+\al)-2 \sigma_n '(t)\right]}.
$$
Substituting either solution into (\ref{com}), we obtain (\ref{sode}) after clearing the square root.
\end{proof}

Although the differential equation (\ref{sode}) for $\s_n(t)$ is very complicated, we find that it will be reduced to a very simple form under a suitable double scaling. Furthermore, we will show that $\s_n(t)$ can be expressed in terms of the Jimbo-Miwa-Okamoto $\s$-function of a particular Painlev\'{e} V in next section; see (\ref{rs1}) and (\ref{rs2}).
\begin{corollary}
If $n\rightarrow\infty$ and $t\rightarrow 0$ in the way such that $\lim\limits_{n\rightarrow\infty}n^4 t=s>0$ and we assume the limit
$$
\s(s):=\lim_{n\rightarrow\infty}\s_n\left(\frac{s}{n^4}\right)
$$
exists. Then $\sigma(s)$ satisfies the following second-order differential equation
\be\label{sf}
4s^2 (\sigma''(s))^2+4s\sigma'(s)\sigma''(s)+8s(\sigma'(s))^3-4\sigma(s)(\sigma'(s))^2+(\sigma'(s))^2=0.
\ee
\end{corollary}
\begin{proof}
With the change of variables $t\mapsto s/n^{4}, \s_n(t)\mapsto \s(s)$, equation (\ref{sode}) turns into
\bea
&&s\left[4s^2 (\sigma''(s))^2+4s\sigma'(s)\sigma''(s)+8s(\sigma'(s))^3-4\sigma(s)(\sigma'(s))^2+(\sigma'(s))^2\right]^2\nonumber\\
&+&\frac{8(-1)^n}{n^3}\Big\{8 s^4 (\sigma ''(s))^3-4 s^2 \left[6 s^2 (\sigma '(s))^2-12 s \sigma (s) \sigma '(s)-3 s \sigma '(s)+4 \sigma^2 (s)\right](\sigma ''(s))^2\nonumber\\
&-&2 s \sigma '(s) \left[4 s^2 (\sigma '(s))^2-20 s \sigma (s) \sigma '(s)-3 s \sigma '(s)+8 \sigma^2 (s)\right]\sigma ''(s)-48 s^3 (\sigma '(s))^5\nonumber\\
&+&2 s^2\left(60 \sigma (s)+1\right) (\sigma '(s))^4-s\left(80 \sigma^2 (s)-8 \sigma (s)-1\right) (\sigma '(s))^3\nonumber\\
&+&4 \sigma^2 (s)(4 \sigma (s)-1) (\sigma '(s))^2\Big\}+O\left(\frac{1}{n^4}\right)=0.\nonumber
\eea
Let $n\rightarrow\infty$, then we obtain (\ref{sf}) by retaining only the highest order terms.
\end{proof}

\begin{remark}
We also find the second-order differential equation (\ref{sf}) in the singular perturbed Gaussian weight problem under a suitable double scaling, and we call this equation the $\s$-form of Painlev\'{e} III\:$'$; see Theorem 3.3 in \cite{MLC}. This phenomenon is termed universality in random matrix theory \cite{Deift,Kui}. In addition, Forrester and Witte \cite{FW2002,FW2006} discovered that the generating function for the probability that an interval contains $k$ eigenvalues in the hard edge scaling limit of the Laguerre unitary ensemble can be evaluated in terms of another $\s$-form of Painlev\'{e} III\:$'$.
\end{remark}

\begin{theorem}
The quantity $\s_n(t)$ satisfies the following second-order nonlinear difference equation:
\bea\label{sd}
&&\left\{n-\frac{(\sigma _{n-1}-\sigma _n)(\sigma _n-\sigma _{n+1})\left[2 \alpha  (n-t)+2 (-1)^n t (n+\alpha)+n^2+\s_n\right]}{2 \left[(-1)^n t (2n+1+2\al)(2n-1+2\al)+(n+\alpha) (\sigma _{n-1}-\sigma _n) (\sigma _n-\sigma _{n+1})\right]}\right\}\nonumber\\[10pt]
&\times&\left\{n+2\al-\frac{(\sigma _{n-1}-\sigma _n)(\sigma _n-\sigma _{n+1})\left[2 \alpha  (n-t)+2 (-1)^n t (n+\alpha)+n^2+\s_n\right]}{2 \left[(-1)^n t (2n+1+2\al)(2n-1+2\al)+(n+\alpha) (\sigma _{n-1}-\sigma _n) (\sigma _n-\sigma _{n+1})\right]}\right\}\nonumber\\[10pt]
&=& t(2n-1+2\al+\sigma _{n-1}-\sigma _n)(2n+1+2\al+\sigma _{n}-\sigma _{n+1})\nonumber\\[10pt]
&\times&\frac{2 \alpha  (n-t)+2 (-1)^n t (n+\alpha)+n^2+\s_n}{ t (2n+1+2\al)(2n-1+2\al)+(-1)^n(n+\alpha) (\sigma _{n-1}-\sigma _n) (\sigma _n-\sigma _{n+1})},
\eea
with the initial conditions
$$
\s_1(t)=-R_0(t),\qquad \s_2(t)=-(R_0(t)+R_1(t))
$$
and $R_0(t)$ and $R_1(t)$ are evaluated in (\ref{ic1}) and (\ref{ic2}).
\end{theorem}
\begin{proof}
From (\ref{sn}) we have
$$
R_n(t)=\s_n(t)-\s_{n+1}(t).
$$
Then (\ref{s2p1}) and (\ref{be3}) becomes
\be\label{eq3}
-2(-1)^n t\: r_{n}(t)=\beta_{n}(\s_n(t)-\s_{n+1}(t))(\s_{n-1}(t)-\s_{n}(t))
\ee
and
\be\label{eq4}
\bt_n=\frac{(n+r_{n}(t))^2+2\al(n+r_{n}(t))}{(2n-1+2\al)(2n+1+2\al+\s_n(t)-\s_{n+1}(t))}+\frac{2(-1)^n t\: r_{n}(t)}{(2n-1+2\al)(\s_n(t)-\s_{n+1}(t))},
\ee
respectively.
Substituting (\ref{be4}) into (\ref{eq3}), we get a linear equation for $r_n(t)$:
\bea
&&-2(-1)^n(2n+1+2\al)(2n-1+2\al) t\: r_{n}(t)\nonumber\\
&=&\left[n(n+2\al)+2(n+\al)r_n(t)+\s_n(t)+2t\left((-1)^n(n+\al)-\al\right)\right]\nonumber\\
&\times&(\s_n(t)-\s_{n+1}(t))(\s_{n-1}(t)-\s_{n}(t)).\nonumber
\eea
So $r_n(t)$ can be expressed in terms of $\s_n(t), \s_{n+1}(t)$ and $\s_{n-1}(t)$:
\begin{small}
\be\label{rnsn}
r_n(t)=-\frac{(\sigma _{n-1}(t)-\sigma _n(t))(\sigma _n(t)-\sigma _{n+1}(t))\left[2 \alpha  (n-t)+2 (-1)^n t (n+\alpha)+n^2+\s_n(t)\right]}{2 \left[(-1)^n t (2n+1+2\al)(2n-1+2\al)+(n+\alpha) (\sigma _{n-1}(t)-\sigma _n(t)) (\sigma _n(t)-\sigma _{n+1}(t))\right]}.
\ee
\end{small}
Substituting (\ref{eq4}) into (\ref{eq3}), and using (\ref{rnsn}) to eliminate $r_n(t)$, we obtain (\ref{sd}). The initial conditions come from formula (\ref{sn}).
\end{proof}

In the next theorem, we display the integral representation of $D_n(t)$ in terms of $R_n(t)$.
\begin{theorem}
\bea
\ln\frac{ D_n(t)}{D_{n}(0)}&=&\int_{0}^{t}\big[4 s^2 (2 n+1+2 \alpha) (R_n'(s))^2-4s R_n(s) (2 n +1+2 \alpha+R_n(s))R_n'(s)\nonumber\\
&-&( 2 n-1+2 \alpha) R_n^4(s)-2 \left(2 n^2+4 \alpha ( n-  s)-1\right)R_n^3(s)\nonumber\\
&-&(2 n+1+2 \alpha)  (4 s^2-8 \alpha  s-1)R_n^2(s)-8 s^2 (2 n+1+2 \alpha)^2 R_n(s)\nonumber\\
&-&4 s^2 (2 n+1+2 \alpha)^3\big]\frac{ds}{8s R_n^2(s) (2 n +1+2 \alpha+R_n(s))}.\nonumber
\eea
\end{theorem}
\begin{proof}
The combination of (\ref{be3}) and (\ref{be4}) shows that $\s_n(t)$ can be expressed in terms of $R_n(t)$ and $r_n(t)$:
\bea\label{sr}
\s_n(t)&=&-n^2-2\al(n-t)-2(n+\al)((-1)^n t+r_n(t))\nonumber\\
&+&(2n+1+2\al)\left[\frac{2 (-1)^n t\: r_n(t)}{R_n(t)}+\frac{(n+r_n(t)) (n+2 \alpha +r_n(t))}{2 n+1+2 \alpha +R_n(t)}\right].
\eea
Substituting (\ref{rexp}) into (\ref{sr}), we have
\bea
\s_n(t)&=&\frac{1}{4 R_n^2(t) (2 n +1+2 \alpha+R_n(t))}\big[4 t^2 (2 n+1+2 \alpha) (R_n'(t))^2\nonumber\\
&-&4 t R_n(t) (2 n +1+2 \alpha+R_n(t))R_n'(t)-( 2 n-1+2 \alpha) R_n^4(t)\nonumber\\
&-&2 \left(2 n^2+4 \alpha ( n-  t)-1\right)R_n^3(t)-(2 n+1+2 \alpha)  (4 t^2-8 \alpha  t-1)R_n^2(t)\nonumber\\
&-&8 t^2 (2 n+1+2 \alpha)^2 R_n(t)-4 t^2 (2 n+1+2 \alpha)^3\big].\nonumber
\eea
In view of the definition of $\s_n(t)$ in (\ref{def}), we have the desired result. Note that $D_n(0)$ is explicitly given by (\ref{bg}).
\end{proof}
\begin{remark}\label{rep}
In random matrix theory, the integral representation of $D_n(t)$ in terms of $R_n(t)$ is very important in the asymptotic analysis of the Hankel determinant. This is because $R_n(t)$ is usually related to a particular Painlev\'{e} transcendent. Then the asymptotics of $D_n(t)$ can be found via the Painlev\'{e} equations.
\end{remark}

\section{Double Scaling Analysis of the Hankel Determinant}
In this section, we consider the asymptotics of the Hankel determinant under a suitable double scaling. As is mentioned in Remark \ref{rep}, the asymptotics of $D_n(t)$ can be obtained from the integral representation involving a Painlev\'{e} transcendent. However, we do not intend to adopt this method, since we can establish the relation of our problem with a previous study by Chen and Dai.

In the paper of Chen and Dai \cite{ChenDai}, they studied the Hankel determinant generated by a Pollaczek-Jacobi type weight, i.e.,
$$
\tilde{D}_{n}(t,a,b):=\det\left(\int_{0}^{1}x^{i+j}x^{a}(1-x)^{b}\mathrm{e}^{-\frac{t}{x}}dx\right)_{i,j=0}^{n-1}.
$$
Note that we use the parameters $a$ and $b$ instead of $\al$ and $\bt$ in their paper to avoid confusion with our previous parameters.
The orthogonality with the Pollaczek-Jacobi type weight is
$$
\int_{0}^{1}\tilde{P}_{j}(x,a,b)\tilde{P}_{k}(x,a,b)x^{a}(1-x)^{b}\mathrm{e}^{-\frac{t}{x}}dx=\tilde{h}_{j}(t,a,b)\delta_{jk},
$$
where $\tilde{P}_{j}(x,a,b),\; j=0,1,2,\ldots$, are the monic polynomials of degree $j$ orthogonal with respect to the weight $x^{a}(1-x)^{b}\mathrm{e}^{-\frac{t}{x}}$. Obviously, $\tilde{P}_{j}(x,a,b)$ should also depend on $t$, while we do not write it down for simplicity.

Based on the ladder operator approach, Chen and Dai \cite{ChenDai} obtained the following result.
\begin{lemma}\label{cd}
The logarithmic derivative of the Hankel determinant with respect to $t$,
\be\label{def1}
H_{n}(t,a,b):=t \frac{d}{dt}\ln \tilde{D}_{n}(t,a,b),
\ee
satisfies the following nonlinear second-order differential equation,
\be\label{hn}
(t H_{n}'')^2=\left[n(n+a+b)-H_n+(a+t)H_n'\right]^2+4H_n'(tH_n'-H_n)(b-H_n').
\ee
Let
$$
\tilde{H}_{n}(t,a,b):=H_{n}(t,a,b)-n(n+a+b),
$$
then (\ref{hn}) becomes
\bea
(t\tilde{H}_{n}'')^2&=&-4t(\tilde{H}_{n}')^3+(\tilde{H}_{n}')^2\left[4\tilde{H}_{n}+(a+2b+t)^2+4n(n+a+b)-4b(a+b)\right]\nonumber\\
&-&2\tilde{H}_{n}'\left[(a+2b+t)\tilde{H}_{n}+2nb(n+a+b)\right]+\tilde{H}_{n}^2,\nonumber
\eea
which is the Jimbo-Miwa-Okamoto $\s$-form of Painlev\'{e} V, with parameters $\nu_{0}=0,\; \nu_{1}=-(n+a+b),\; \nu_{2}=n,\; \nu_{3}=-b$, following the convention in \cite[(C.45)]{Jimbo1981}.
\end{lemma}

\begin{proposition}
We have
\be\label{de1}
\tilde{P}_j\left(x,-\frac{1}{2},\al\right)=P_{2j}(\sqrt{x}),\qquad j=0,1,2,\ldots,
\ee
\be\label{de2}
\tilde{P}_j\left(x,\frac{1}{2},\al\right)=\frac{P_{2j+1}(\sqrt{x})}{\sqrt{x}},\qquad j=0,1,2,\ldots,
\ee
where $\left\{\tilde{P}_j\left(x,-\frac{1}{2},\al\right)\right\}_{j=0}^{\infty}$ are the monic orthogonal polynomials with respect to the weight function $x^{-\frac{1}{2}}(1-x)^{\al}\mathrm{e}^{-\frac{t}{x}}$, and $\left\{\tilde{P}_j\left(x,\frac{1}{2},\al\right)\right\}_{j=0}^{\infty}$ are the monic orthogonal polynomials with respect to the weight $x^{\frac{1}{2}}(1-x)^{\al}\mathrm{e}^{-\frac{t}{x}}$. Here the subscript $j$ represents the degree of the corresponding polynomial. Furthermore, the following relations hold:
\be\label{rela1}
h_{2j}(t)=\tilde{h}_{j}\left(t,-\frac{1}{2},\alpha\right),
\ee
\be\label{rela2}
h_{2j+1}(t)=\tilde{h}_{j}\left(t,\frac{1}{2},\alpha\right).
\ee

\end{proposition}
\begin{proof}
At first, it is easy to see that the right hand side of (\ref{de1}) and (\ref{de2}) are indeed polynomials of degree of $j$.
From the orthogonality (\ref{or}), we have
\bea
h_{2j}(t)\delta_{2j,2k}&=&\int_{-1}^{1}P_{2j}(x)P_{2k}(x)(1-x^2)^\al\mathrm{e}^{-\frac{t}{x^{2}}}dx\nonumber\\
&=&2\int_{0}^{1}P_{2j}(x)P_{2k}(x)(1-x^2)^\al\mathrm{e}^{-\frac{t}{x^{2}}}dx\nonumber\\
&=&\int_{0}^{1}P_{2j}(\sqrt{y})P_{2k}(\sqrt{y})y^{-\frac{1}{2}}(1-y)^{\al}\mathrm{e}^{-\frac{t}{y}}dy,\qquad j,k=0,1,2,\ldots,\nonumber
\eea
which implies
$$
\tilde{P}_j\left(y,-\frac{1}{2},\al\right)=P_{2j}(\sqrt{y}),\qquad j=0,1,2,\ldots.
$$
Then
\bea
h_{2j}(t)
&=&\int_{0}^{1}\left[\tilde{P}_j\left(y,-\frac{1}{2},\al\right)\right]^{2}y^{-\frac{1}{2}}(1-y)^{\al}\mathrm{e}^{-\frac{t}{y}}dy\nonumber\\
&=&\tilde{h}_{j}\left(t,-\frac{1}{2},\alpha\right),\qquad j=0,1,2,\ldots.\nonumber
\eea
On the other hand,
\bea
h_{2j+1}(t)\delta_{2j+1,2k+1}&=&\int_{-1}^{1}P_{2j+1}(x)P_{2k+1}(x)(1-x^2)^\al\mathrm{e}^{-\frac{t}{x^{2}}}dx\nonumber\\
&=&2\int_{0}^{1}P_{2j+1}(x)P_{2k+1}(x)(1-x^2)^\al\mathrm{e}^{-\frac{t}{x^{2}}}dx\nonumber\\
&=&\int_{0}^{1}\frac{P_{2j+1}(\sqrt{y})}{\sqrt{y}}\frac{P_{2k+1}(\sqrt{y})}{\sqrt{y}}y^{\frac{1}{2}}(1-y)^{\al}\mathrm{e}^{-\frac{t}{y}}dy.\nonumber
\eea
It follows that
$$
\tilde{P}_j\left(y,\frac{1}{2},\al\right)=\frac{P_{2j+1}(\sqrt{y})}{\sqrt{y}},\qquad j=0,1,2,\ldots.
$$
Hence,
\bea
h_{2j+1}(t)
&=&\int_{0}^{1}\left[\tilde{P}_j\left(y,\frac{1}{2},\al\right)\right]^{2}y^{\frac{1}{2}}(1-y)^{\al}\mathrm{e}^{-\frac{t}{y}}dy\nonumber\\
&=&\tilde{h}_{j}\left(t,\frac{1}{2},\alpha\right),\qquad j=0,1,2,\ldots.\nonumber
\eea
The proof is complete.
\end{proof}
From (\ref{rela1}) and (\ref{rela2}), and with the aid of (\ref{hankel}), we establish the relation between $D_{n}(t)$ and $\tilde{D}_{n}(t,a,b)$, with $n=0,1,2,\ldots$,
\be\label{hd1}
D_{2n}(t)=\tilde{D}_{n}\left(t,\frac{1}{2},\al\right)\tilde{D}_{n}\left(t,-\frac{1}{2},\al\right),
\ee
\be\label{hd2}
D_{2n+1}(t)=\tilde{D}_{n}\left(t,\frac{1}{2},\al\right)\tilde{D}_{n+1}\left(t,-\frac{1}{2},\al\right).
\ee
In addition, in view of (\ref{def}) and (\ref{def1}), we find for $n=0,1,2,\ldots,$
\be\label{re3}
\sigma_{2n}(t)=2\left[H_{n}\left(t,\frac{1}{2},\alpha\right)+H_{n}\left(t,-\frac{1}{2},\alpha\right)\right],
\ee
\be\label{re4}
\sigma_{2n+1}(t)=2\left[H_{n}\left(t,\frac{1}{2},\alpha\right)+H_{n+1}\left(t,-\frac{1}{2},\alpha\right)\right].
\ee
Hence, we can express $\s_n(t)$ in terms of the $\s$-function of a Painlev\'{e} V:
\be\label{rs1}
\sigma_{2n}(t)=2\left[\tilde{H}_{n}\left(t,\frac{1}{2},\alpha\right)+\tilde{H}_{n}\left(t,-\frac{1}{2},\alpha\right)+2n(n+\al)\right],
\ee
\be\label{rs2}
\sigma_{2n+1}(t)=2\left[\tilde{H}_{n}\left(t,\frac{1}{2},\alpha\right)+\tilde{H}_{n+1}\left(t,-\frac{1}{2},\alpha\right)+(2n+1)\left(n+\al+\frac{1}{2}\right)\right],
\ee
where the function $\tilde{H}_{n}(t,a,b)$, as shown in Lemma \ref{cd}, satisfies the Jimbo-Miwa-Okamoto $\s$-form of a particular Painlev\'{e} V.

Similarly, according to the definition of $R_{n}(t)$ in (\ref{Rnt}), we obtain
\be\label{dou1}
R_{2n}(t)=2R_{n}^*\left(t,-\frac{1}{2},\al\right)=2\left(\tilde{R}_n\left(t,-\frac{1}{2},\al\right)-2n-\frac{1}{2}-\al\right),
\ee
\be\label{dou2}
R_{2n+1}(t)=2R_{n}^*\left(t,\frac{1}{2},\al\right)=2\left(\tilde{R}_n\left(t,\frac{1}{2},\al\right)-2n-\frac{3}{2}-\al\right),
\ee
for $n=0,1,2,\ldots$, where $R_{n}^*(t,a,b)$ and $\tilde{R}_{n}(t,a,b)$ are defined by (2.16) and (2.17) in \cite{ChenDai},
$$
R_{n}^*(t,a,b):=\frac{t}{\tilde{h}_{n}(t,a,b)}\int_{0}^{1}\frac{\tilde{P}_{n}^2(y,a,b)}{y}y^{a}(1-y)^b\mathrm{e}^{-\frac{t}{y}}dy,
$$
$$
\tilde{R}_{n}(t,a,b):=\frac{b}{\tilde{h}_{n}(t,a,b)}\int_{0}^{1}\frac{\tilde{P}_{n}^2(y,a,b)}{1-y}y^{a}(1-y)^b\mathrm{e}^{-\frac{t}{y}}dy,
$$
and they have a simple relation given by formula (2.25) in \cite{ChenDai},
$$
R_{n}^*(t,a,b)=\tilde{R}_{n}(t,a,b)-2n-1-a-b.
$$

By employing the general linear statistics results from \cite{Chen1998}, based on Dyson's Coulomb fluid interpretation \cite{Dyson}, Chen et al. \cite{CCF} studied the asymptotic behavior of $\tilde{R}_n(t,a,b)$ and $\tilde{D}_{n}(t,a,b)$ under a suitable double scaling. According to their results, we obtain the asymptotics of our scaled $R_n(t)$ and $D_n(t)$.

\begin{theorem}
Assume that $n\rightarrow\infty$ and $t\rightarrow 0$ such that $s=2n^2 t$ is fixed and the following limits
$$
g_1(s):=\lim_{n\rightarrow\infty}n\: R_{2n}\left(\frac{s}{2n^2}\right)
$$
$$
g_2(s):=\lim_{n\rightarrow\infty}n\:R_{2n+1}\left(\frac{s}{2n^2}\right)
$$
exist. Then for $s\rightarrow 0$,
$$
g_1(s)=-4 s+\frac{32 s^2}{3}-\frac{256 s^3}{15}+\frac{8192 s^4}{315}-\frac{311296 s^5}{8505}+\frac{7733248 s^6}{155925}+O(s^7),
$$
$$
g_2(s)=4 s+\frac{32 s^2}{3}+\frac{256 s^3}{15}+\frac{8192 s^4}{315}+\frac{311296 s^5}{8505}+\frac{7733248 s^6}{155925}+O(s^7),
$$
and for $s\rightarrow \infty$,
$$
g_1(s)=2 s^{\frac{2}{3}}+\frac{1}{3}s^{\frac{1}{3}}+\frac{1}{108s^{\frac{1}{3}}}-\frac{1}{648s^{\frac{2}{3}}}+\frac{1}{324 s}-\frac{7}{5832 s^{\frac{4}{3}}}+O(s^{-\frac{5}{3}}),
$$
$$
g_2(s)=2 s^{\frac{2}{3}}-\frac{1}{3}s^{\frac{1}{3}}-\frac{1}{108s^{\frac{1}{3}}}-\frac{1}{648s^{\frac{2}{3}}}-\frac{1}{324 s}-\frac{7}{5832 s^{\frac{4}{3}}}+O(s^{-\frac{5}{3}}).
$$
\end{theorem}
\begin{proof}
According to (\ref{dou1}) and (\ref{dou2}), we have
\be\label{re11}
g_1(s)=4g\left(s,-\frac{1}{2},\al\right),
\ee
\be\label{re12}
g_2(s)=4g\left(s,\frac{1}{2},\al\right),
\ee
where $g(s,a,b)$ is defined by \cite[p.6]{CCF}
$$
g(s,a,b)=\lim_{n\rightarrow\infty}n^2\left(\frac{\tilde{R}_{n}\left(\frac{s}{2n^2},a,b\right)}{2n+1+a+b}-1\right).
$$
and $g(s,a,b)$ satisfies a particular Painlev\'{e} III$'$ \cite[(A.45.4$'$)]{Mehta}:
$$
g''=\frac{(g')^2}{g}-\frac{g'}{s}+\frac{2g^2}{s^2}+\frac{a}{2s}-\frac{1}{4g},
$$
i.e., $P_{III'}(8,2a,0,-1)$, with the initial conditions $g(0,a,b)=0, g'(0,a,b)=\frac{1}{2a}$. Then Chen et al. obtained the small $s$ and large $s$ asymptotics of $g(s,a,b)$, which are given by formulas (2.10) and (2.11) in \cite{CCF}:
\bea
g(s,a,b)&=&\frac{s}{2a}-\frac{s^2}{2a^2(a^2-1)}+\frac{3s^3}{2a^3(a^2-1)(a^2-4)}+\frac{3(3-2 a^2) s^4}{a^4 (a^2-1)^2(a^2-4)(a^2-9) }\nonumber\\[10pt]
&+&\frac{5 (11 a^2-36) s^5}{2 a^5(a^2-1)^2(a^2-4)  (a^2-9)  (a^2-16)}\nonumber\\[10pt]
&-&\frac{3 (91 a^6-1115 a^4+4219 a^2-3600) s^6}{2 a^6 (a^2-1)^3(a^2-4)^2 (a^2-9)(a^2-16)(a^2-25)}
+O(s^7),\quad s\rightarrow 0,\quad a\notin \mathbb{Z},\nonumber
\eea
\bea
g(s,a,b)&=&\frac{1}{2}s^{\frac{2}{3}}-\frac{a}{6}s^{\frac{1}{3}}+\frac{a(a^2-1)}{162s^{\frac{1}{3}}}+\frac{a^2(a^2-1)}{486s^{\frac{2}{3}}}+\frac{a(a^2-1)}{486 s}-\frac{a^2(a^2-1)(2a^2-11)}{6561s^{\frac{4}{3}}}\nonumber\\
&+&O(s^{-\frac{5}{3}}),\qquad s\rightarrow \infty.\nonumber
\eea
Here we produce more terms in the asymptotic expansions by using their method. In view of the relations (\ref{re11}) and (\ref{re12}), we obtain the desired results of the theorem.
\end{proof}

\begin{theorem}\label{thm}
Suppose that $n\rightarrow\infty$ and $t\rightarrow 0$ such that $s=2n^2 t$ is fixed and the following limits
$$
\Delta_{1}(s):=\lim_{n\rightarrow\infty}\frac{D_{2n}\left(\frac{s}{2n^2}\right)}{D_{2n}(0)},
$$
$$
\Delta_{2}(s):=\lim_{n\rightarrow\infty}\frac{D_{2n+1}\left(\frac{s}{2n^2}\right)}{D_{2n+1}(0)}
$$
exist. We have for $s\rightarrow 0$,
$$
\Delta_{1}(s)=\Delta_{2}(s)=\exp\left(-\frac{4 s^2}{3}-\frac{256 s^4}{315}-\frac{966656 s^6}{1403325}+O(s^8)\right),
$$
while for $s\rightarrow\infty$,
$$
\Delta_{1}(s)=\Delta_{2}(s)=\exp\left(\frac{\ln 2}{12}+3\zeta'(-1)-\frac{\ln s}{36}-\frac{9}{4} s^{\frac{2}{3}}+\frac{1}{576}s^{-\frac{2}{3}}+\frac{7}{20736}s^{-\frac{4}{3}}+O(s^{-2})\right),
$$
where  $\zeta(\cdot)$ is the Riemann zeta function. Here the constant term $\frac{\ln 2}{12}+3\zeta'(-1)$ is called Dyson's constant \cite{Dyson1976}.
\end{theorem}
\begin{proof}
From (\ref{hd1}) and (\ref{hd2}) we have
\be\label{re13}
\Delta_{1}(s)=\Delta_{2}(s)=\Delta\left(s,\frac{1}{2},\al\right)\Delta\left(s,-\frac{1}{2},\al\right),
\ee
where
$$
\Delta(s,a,b):=\lim_{n\rightarrow\infty}\frac{\tilde{D}_{n}\left(\frac{s}{2n^2},a,b\right)}{\tilde{D}_{n}(0,a,b)}.
$$
By using the integral representation of $\Delta(s,a,b)$ in terms of $g(s,a,b)$, Chen et al. obtained the small $s$ and large $s$ asymptotics of $\Delta(s,a,b)$ in \cite[Thm. 4]{CCF}: for $s\rightarrow 0$,
\bea
\Delta(s,a,b)&=&\exp\Big[-\frac{s}{2 a}+\frac{s^2}{8 a^2 (a^2-1)}-\frac{s^3}{6 a^3(a^2-1)(a^2-4)}\nonumber\\[10pt]
&+&\frac{3 (2 a^2-3) s^4}{16 a^4 (a^2-1)^2 (a^2-4)(a^2-9) }-\frac{(11 a^2-36)s^5}{10 a^5 (a^2-1)^2  (a^2-4)(a^2-9) (a^2-16)}\nonumber\\[10pt]
&+&\frac{(91 a^6-1115 a^4+4219 a^2-3600)s^6}{24 a^6 (a^2-1)^3 (a^2-4)^2(a^2-9)(a^2-16)(a^2-25) }+O(s^7)\Big],\nonumber
\eea
and for $s\rightarrow\infty$,
\bea
\Delta(s,a,b)&=&\exp\Big[c(a)-\frac{9}{8}s^{\frac{2}{3}}+\frac{3a}{2}s^{\frac{1}{3}}+\frac{1-6a^2}{36}\ln s-\frac{a(a^2-1)}{18}s^{-\frac{1}{3}}-\frac{a^2(a^2-1)}{216}s^{-\frac{2}{3}}\nonumber\\[10pt]
&-&\frac{a(a^2-1)}{486}s^{-1}+\frac{a^2(a^2-1)(2a^2-11)}{11664}s^{-\frac{4}{3}}+\frac{a(a^2-1)(a^4-a^2-15)}{21870}s^{-\frac{5}{3}}\nonumber\\[10pt]
&+&O(s^{-2})\Big],\nonumber
\eea
where $c(a)$ is an integration constant, independent of $s$, given by
\be\label{cons}
c(a)=\ln\frac{G(a+1)}{(2\pi)^{\frac{a}{2}}},
\ee
and $G(\cdot)$ is the Barnes $G$-function. Here we show more terms in the asymptotic expansions of $\Delta(s,a,b)$.

In view of (\ref{re13}), we find for $s\rightarrow 0$,
$$
\Delta_{1}(s)=\Delta_{2}(s)=\exp\left(-\frac{4 s^2}{3}-\frac{256 s^4}{315}-\frac{966656 s^6}{1403325}+O(s^8)\right),
$$
and for $s\rightarrow\infty$,
$$
\Delta_{1}(s)=\Delta_{2}(s)=\exp\left(c\left(\frac{1}{2}\right)+c\left(-\frac{1}{2}\right)-\frac{\ln s}{36}-\frac{9}{4} s^{\frac{2}{3}}+\frac{1}{576}s^{-\frac{2}{3}}+\frac{7}{20736}s^{-\frac{4}{3}}+O(s^{-2})\right).
$$
It follows from (\ref{cons}) that
$$
c\left(\frac{1}{2}\right)+c\left(-\frac{1}{2}\right)=\ln \left(G\left(\frac{1}{2}\right)G\left(\frac{3}{2}\right)\right)=\ln \left(G^2\left(\frac{1}{2}\right)\Gamma\left(\frac{1}{2}\right)\right).
$$
With the aid of the expression of $G\left(\frac{1}{2}\right)$ in \cite[(6.39)]{Voros}, i.e.,
$$
G\left(\frac{1}{2}\right)=2^{\frac{1}{24}}\pi^{-\frac{1}{4}}\mathrm{e}^{\frac{3\zeta'(-1)}{2}},
$$
and $\zeta(\cdot)$ is the Riemann zeta function, we have
$$
c\left(\frac{1}{2}\right)+c\left(-\frac{1}{2}\right)=\frac{\ln 2}{12}+3\zeta'(-1).
$$
Hence, the theorem is established.
\end{proof}
\begin{remark}
We find that the small $s$ and large $s$ asymptotic behaviors of the scaled Hankel determinant in Theorem \ref{thm} are identical with the asymptotics of the Hankel determinant for the singularly perturbed Gaussian weight under a different double scaling \cite{MLC}. Our results demonstrate the universality in random matrix theory. That is, although the Hankel determinants for different weights are quite distinct in the finite dimensional case, they may have the same characteristics when considering the limiting behavior.
\end{remark}

\section*{Acknowledgments}
Chao Min was supported by the Scientific Research Funds of Huaqiao University under grant number 17BS402.
Yang Chen was supported by the Macau Science and Technology Development Fund under grant number FDCT 023/2017/A1 and by the University of Macau under grant number MYRG 2018-00125-FST.

\newpage
\noindent School of Mathematical Sciences, Huaqiao University, Quanzhou 362021, China\\
e-mail: chaomin@hqu.edu.cn\\

\noindent Department of Mathematics, Faculty of Science and Technology, University of Macau, Macau, China\\
e-mail: yangbrookchen@yahoo.co.uk

\end{document}